\documentclass[12pt]{iopart}

\usepackage{amsmath, amssymb}
\usepackage{graphicx}
\usepackage{geometry}
\usepackage{bm}
\usepackage{arydshln}
\usepackage{braket, mathtools}
\usepackage{amsthm}
\usepackage{color, hyperref}
\usepackage{cite}
\usepackage{cancel}
\usepackage{multirow}
\usepackage{comment}

\theoremstyle{plain}
\newtheorem{theo}{Theorem}[section]

\newtheorem{prop}[theo]{Proposition}

\theoremstyle{definition}
\newtheorem{rema}[theo]{Remark}
\newtheorem{exam}[theo]{Example}

\newcommand{\cT}{\mathcal{T}}
\newcommand{\cH}{\mathcal{H}}
\newcommand{\cV}{\mathcal{V}}
\newcommand{\cK}{\mathcal{K}}

\newcommand{\ketbra}[2]{\ket{#1}\!\bra{#2}}
\newcommand{\abs}[1]{\left|#1\right|}
\renewcommand{\Re}{\operatorname{Re}}

\begin{document}

\title[Perfect Discrimination of Non-Orthogonal Separable Pure States]{Perfect Discrimination of Non-Orthogonal Separable Pure States on Bipartite System in General Probabilistic Theory}

\author{Hayato Arai$^1$, Yuuya Yoshida$^1$, and Masahito Hayashi$^{1,2,3}$}

\address{$^1$ Graduate School of Mathematics, Nagoya University, Nagoya, Japan}
\address{$^2$ Shenzhen Institute for Quantum Science and Engineering, Southern University
of Science and Technology, Nanshan District, Shenzhen 518055,
People's Republic of China}
\address{$^3$ Centre for Quantum Technologies, National University of Singapore, 3 Science
Drive 2, 117542, Singapore}
\eads{\mailto{m18003b@math.nagoya-u.ac.jp}, \mailto{m17043e@math.nagoya-u.ac.jp}, \mailto{masahito@math.nagoya-u.ac.jp}}
\vspace{10pt}

\begin{abstract}
We address perfect discrimination of two separable states. 
When available states are restricted to separable states, 
we can theoretically consider a larger class of measurements than the class of measurements allowed in quantum theory. 
The framework composed of the class of separable states and the above extended class of measurements 
is a typical example of general probabilistic theories. 
In this framework, 
we give a necessary and sufficient condition to discriminate two separable pure states perfectly.
In particular, we derive measurements explicitly to discriminate two separable pure states perfectly, 
and find that some non-orthogonal states are perfectly distinguishable. 
However, the above framework does not improve the \textit{capacity}, namely, 
the maximum number of states that are simultaneously and perfectly distinguishable.
\end{abstract}

\vspace{2pc}

\maketitle

\section{Introduction}\label{sect1}

Entanglement is a resource for miracle performance of quantum information processing \cite{Bennett.et.al:1993,Bennett.Wiesner:1992}.
Even when a quantum state has no entanglement, 
entanglement in a measuring process brings us performance that measuring processes without quantum correlation cannot realize.
In fact, when we discriminate the $n$-fold tensor products of two quantum states, 
the performance of measurements with quantum correlation is beyond that of any measurement without quantum correlation, 
e.g., local operation and classical communication (LOCC) and separable measurement \cite{Hiai.Petz:1991,Ogawa.Nagaoka:2000,Audenaert.et.al:2007,Nussbaum.Szkola:2009,HayashiBook:2017}. 
The difference between the first and second performance can be derived from the following two classes of measurements. 
One is the class of measurements allowed in quantum theory and the other is the class of measurements with only separable form.
The first class achieves strictly better performance than the second class in the above discrimination.

All the above studies of state discrimination considered classes of measurements allowed in quantum theory, 
but there is a theoretical possibility that 
a larger class of measurements brings us more miracle performance of state discrimination than that of quantum theory. 
In order to consider a larger class of measurements, 
we need to restrict available states. 
Such a framework is discussed in general probabilistic theories (GPTs) 
\cite{Kimura.et.al:2010,Kimura.et.al:2016,Muller.et.al:2012,Janotta:2013,Janotta:2014,Barnum.et.al:2006,Barnum.et.al:2008,Barnum:2010,Dahlsten.et.al:2012,Lami.et.al:2017,Hamamura:2018,Yoshida:2018,Aubrun.et.al:2018,Matsumoto:2018,Short.et.al:2010,Bae.et.al:2016}, 
which are a generalization of quantum theory and classical probability theory. 
GPTs are the most general framework to characterize states, measurements, and time evolution. 
Although some preceding studies compared GPTs with quantum theory 
\cite{Barnum.et.al:2006,Dahlsten.et.al:2012,Lami.et.al:2017,Barnum:2010,Janotta:2013,Hamamura:2018}, 
few studies clarified the difference between quantum theory and other GPTs in the viewpoint of state discrimination. 
Hence, to clarify the difference, we focus on the following typical GPT on a bipartite system: 
we restrict available states to separable states on the composite system and 
this restriction allows us to consider theoretically measurements that are not allowed in quantum theory. 
The framework composed of the class of separable states and the class of such measurements is a typical example of GPTs and is denoted by SEP.

The difference between quantum theory and SEP
can be characterized by the relation between 
the \textit{positive and dual cones} appeared in quantum theory and SEP, as 
illustrated in figure~\ref{fig:cone}. 
A \textit{positive cone} defines the set of all states in a GPT 
so that a state is given as an element of a positive cone whose trace is one. 
For example, the positive cone of quantum theory is the set of all positive semi-definite matrices 
and the positive cone of SEP is the set of all matrices with separable form. 
Thus, states in SEP are restricted to separable states, and   
the positive cone of SEP is smaller than that of quantum theory. 
This restriction makes bit commitment possible under SEP \cite{Barnum.et.al:2008}.
Furthermore, the \textit{dual cone} of a positive cone defines measurements of a GPT 
so that a measurement is given as a decomposition $\{M_i\}_i$ of the identity matrix $I$. 
More precisely, all elements $M_i$ lie in the dual cone and satisfy $\sum_i M_i=I$. 
For example, the dual cone of quantum theory is also the set of all positive semi-definite matrices 
and the dual cone of SEP is the set of all matrices $Y$ 
that satisfy $\Tr XY\ge0$ for all matrices $X$ with separable form. 
Thus the dual cone of SEP is larger than that of quantum theory. 
Therefore, measurements of SEP contain not only those of quantum theory 
but also those that quantum theory cannot realize.

\begin{figure}[t]
	\centering
	\includegraphics[scale=0.3]{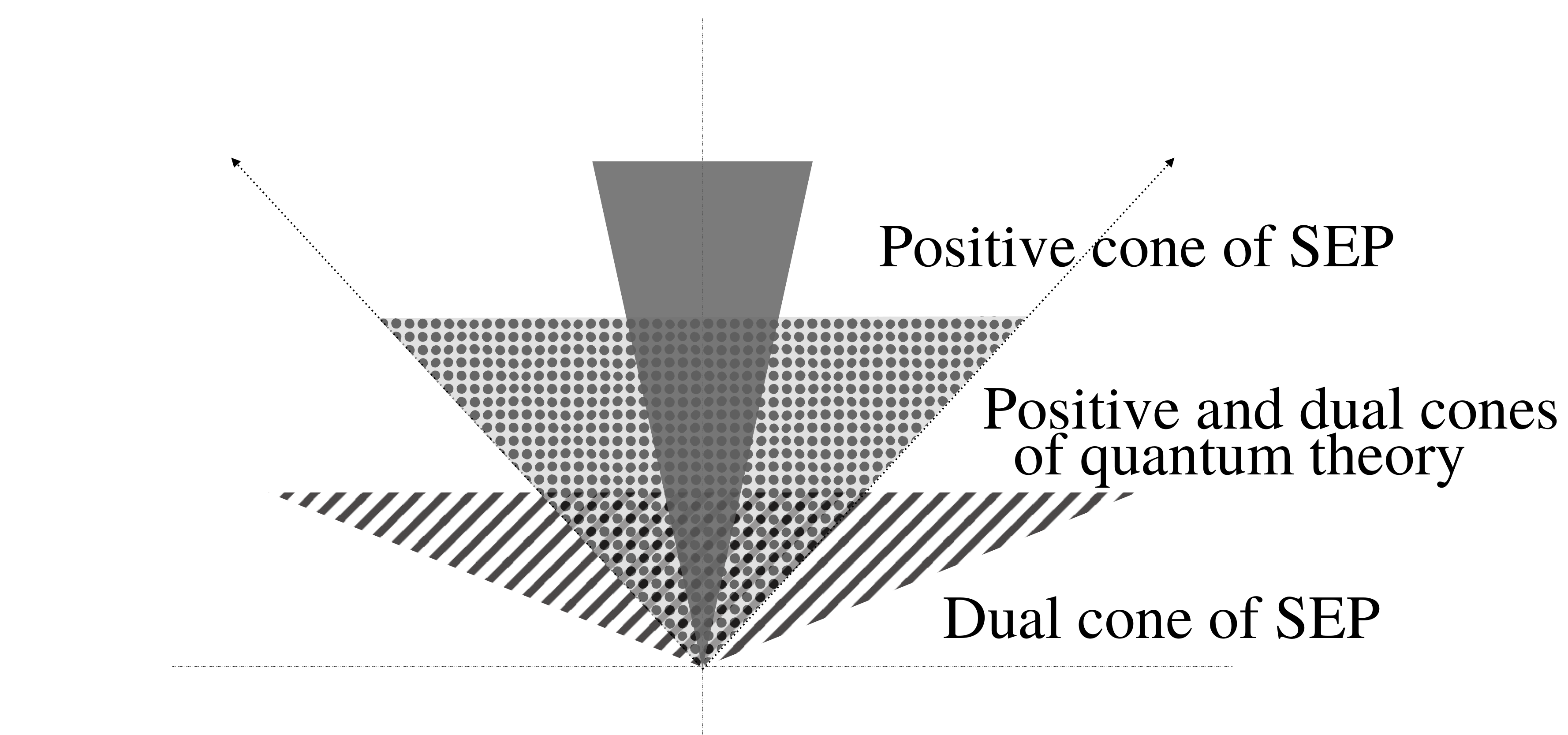}
	\caption{A sketch of the inclusion relation of positive and dual cones.}
	\label{fig:cone}
\end{figure}

In this paper, we address perfect discrimination of two pure states in SEP.
A main goal of this paper is to reveal how much better
the performance of perfect discrimination in SEP is than that in quantum theory. 
In quantum theory, it is well-known that 
orthogonality of two states is necessary and sufficient to discriminate two states perfectly \cite{Nielsen.Chuang.Book:2000}.
This fact is not changed even if we restrect the class of measurements to LOCC \cite{Walgate.et.al:2000}.
However, as shown in this paper, there exists a non-orthogonal pair of two separable pure states that can be discriminated in SEP.
Moreover, we derive a necessary and sufficient condition for state discrimination in SEP.
The necessary and sufficient condition implies that $2n$-copies $\rho_1^{\otimes 2n}$ and $\rho_2^{\otimes 2n}$ of pure states
are perfectly distinguishable for a sufficiently large $n$ if $\rho_1\not=\rho_2$.  
In this sense, SEP is completely different from quantum theory.

Since our necessary and sufficient condition reveals that some non-orthogonal states in SEP can be discriminated perfectly,
one might think that the \textit{capacity} in SEP is improved 
in comparison with the capacity in quantum theory. 
Here the capacity in a GPT is the maximum number of states 
that are simultaneously and perfectly distinguishable in the GPT, 
and expresses the limit of communication quantity 
per single use of quantum communication. 
The capacity in quantum theory is equal to the dimension of a quantum system,
and an interesting relation for the capacities in GPTs has been derived \cite{Muller.et.al:2012}. 
Using the relation \cite[lemma~24]{Muller.et.al:2012}, 
we find that the capacity in SEP is equal to that in quantum theory.


The remaining of this paper is organized as follows.
The beginning of section~\ref{sect2} formulates our extended class of measurements 
and gives a perfectly distinguishable pair of two separable pure states that are not orthogonal.
The latter of section~\ref{sect2} gives a necessary and sufficient condition 
to discriminate two separable pure states in SEP perfectly (theorem~\ref{theo:Main}). 
Also, the latter of section~\ref{sect2} discusses the capacity in SEP (theorem~\ref{theo:capa}). 
Section~\ref{sect3} proves the sufficiency of theorem~\ref{theo:Main} and 
section~\ref{sect4} does the necessity of theorem~\ref{theo:Main}. 
Section~\ref{sect5} is devoted to further discussion.

\section{Perfectly distinguishable pairs of two pure states in SEP}\label{sect2}

First, let us describe our framework SEP and notational conventions. 
Let $\cH_A$ and $\cH_B$ be two finite-dimensional complex Hilbert spaces. 
We denote by $\cT(AB)$ and $\cT_+(AB)$ the set of all Hermitian matrices on $\cH_A\otimes\cH_B$ 
and the set of all positive semi-definite matrices on $\cH_A\otimes\cH_B$, respectively. 
The sets $\cT(A)$, $\cT(B)$, $\cT_+(A)$, and $\cT_+(B)$ are defined similarly. 
In quantum theory, available states are elements of $\cT_+(AB)$ with trace one. 
However, in this paper we look at the scenario 
where the only available states are separable states 
we restrict available states to separable states, i.e., elements of 
\[
\mathrm{SEP}(A;B)
\coloneqq \Set{\sum_i X_i^A\otimes X_i^B | X_i^A\in\cT_+(A),\ X_i^B\in\cT_+(B)\ (\forall i)}
\]
with trace one. 
In order to address state discrimination, 
we must also define measurements of SEP. 
In quantum theory, measurements are given as positive-operator valued measures (POVMs). 
That is, a measurement $\{M_i\}_i$ satisfies $M_i\in\cT_+(AB)$ and $\sum_i M_i=I$ for any outcome $i$. 
However, since we restrict available states to separable states, 
measurements of SEP form a larger class than those of quantum theory. 
A measurement $\{M_i\}_i$ of SEP is defined by the conditions 
\[
M_i\in\mathrm{SEP}^*(A;B)\ (\forall i),\quad \sum_i M_i=I,
\]
where $\mathrm{SEP}^*(A;B)$ denotes the dual cone of $\mathrm{SEP}(A;B)$ and is defined as 
\[
\mathrm{SEP}^*(A;B)
= \set{Y\in\cT(AB) | \Tr XY\ge0\ (\forall X\in\mathrm{SEP}(A;B))}.
\]
Since the inclusion relation $\mathrm{SEP}^*(A;B)\subset\cT_+(AB)$ holds, 
measurements of SEP form a larger class than those of quantum theory.

\begin{rema}
For readers' convenience, we describe SEP again according to GPTs. 
Let $\cV$ be a finite-dimensional real vector space with an inner product $\langle\cdot,\cdot\rangle$. 
We say that $\cK\subset\cV$ is a (proper) \textit{positive cone} 
if $\cK$ is a closed convex set satisfying the following conditions: 
\begin{itemize}
	\item
	$\alpha x\in\cK$ for all $\alpha\ge0$ and $x\in\cK$,
	\item
	$\cK\cap(-\cK)=\{0\}$,
	\item
	The interior of $\cK$ is non-empty.
\end{itemize}
Also, the \textit{dual cone} $\cK^\ast$ of a positive cone $\cK$ is defined as 
\[
\cK^\ast = \set{y\in\cV | \langle x,y\rangle \ge0 \ (\forall x\in\cK)}.
\]
A GPT consists of a real vector space $\cV$, a positive cone $\cK$, 
and an element $u$ of the interior of $\cK^\ast$. 
A \textit{state} $\rho$ of a GPT $(\cV, \cK, u)$ is given as 
an element of $\cK$ satisfying $\langle \rho, u\rangle=1$. 
Also, a \textit{measurement} $\{m_i\}_{i=1}^n$ of the GPT is given as 
a family $\{m_i\}_{i=1}^n$ composed of elements in $\cK^\ast$ satisfying $\sum_{i=1}^n m_i =u$. 
As a framework of states and measurements, 
quantum theory is equivalent to the GPT $(\cT(A), \cT_+(A), I)$, 
and our framework SEP is equivalent to the GPT $(\cT(AB), \mathrm{SEP}(A;B), I)$. 
Also, SEP $(\cT(AB), \mathrm{SEP}(A;B), I)$ is a composite system of two quantum subsystems 
$(\cT(A), \cT_+(A), I)$ and $(\cT(B), \cT_+(B), I)$. 
Since the dual cone $\mathrm{SEP}^\ast$ includes any entanglement witness, 
the framework SEP is often called \textit{witness theory} \cite{Lami.et.al:2017}. 
Moreover, the positive cone SEP is the smallest cone of all positive cones of composite systems of two quantum subsystems, 
and thus it is called the \textit{minimal tensor product} \cite{Janotta:2014}.
\end{rema}

Now, let us consider state discrimination in SEP. 
Let $\{\rho_i\}_{i=1}^n$ be a family of $n$ states. 
Then we say that $\{\rho_i\}_{i=1}^n$ is \textit{perfectly distinguishable} in SEP (resp.\ quantum theory) 
if there exists a measurement $\{M_j\}_{j=1}^n$ of SEP (resp.\ quantum theory) 
such that $\Tr M_j\rho_i=\delta_{ij}$, where $\delta_{ij}$ denotes the Kronecker delta. 
It is well-known that $\{\rho_i\}_{i=1}^n$ is perfectly distinguishable in quantum theory 
if and only if any two distinct states of $\{\rho_i\}_{i=1}^n$ are orthogonal, i.e., 
$\Tr\rho_i\rho_j=\delta_{ij}$ for all $i\not=j$. 
In this paper, we address the case $n=2$ mainly. 

Example~\ref{exam} gives an example that two states are perfectly distinguishable and not orthogonal. 
For this purpose, we consider the case where $\cH_A$ and $\cH_B$ are two-dimensional 
(hereinafter, it is called the $(2,2)$-dimensional case). 
In this case, the dual cone $\mathrm{SEP}^*(A;B)$ 
can be expressed explicitly by using the partial transpose operation $\Gamma$, 
which throughout the paper we assume to be on subsystem $B$. 
Since for $2\times2$ matrices $C=(c_{ij})_{i,j}$ and $D=(d_{ij})_{i,j}$ 
the tensor product matrix $C\otimes D$ is expressed as 
\[
C\otimes D =
\left[
\begin{array}{cc:cc}
	c_{11}d_{11} & c_{11}d_{12} & c_{12}d_{11} & c_{12}d_{12}\\
	c_{11}d_{21} & c_{11}d_{22} & c_{12}d_{21} & c_{12}d_{22}\\ \hdashline
	c_{21}d_{11} & c_{21}d_{12} & c_{22}d_{11} & c_{22}d_{12}\\
	c_{21}d_{11} & c_{21}d_{12} & c_{22}d_{11} & c_{22}d_{12}
\end{array}
\right],
\]
the partial transpose $\Gamma(X)$ of a matrix $X=(x_{ij})_{i,j}$ is 
\[
\Gamma(X) =
\left[
\begin{array}{cc:cc}
	x_{11} & x_{21} & x_{13} & x_{23}\\
	x_{12} & x_{22} & x_{14} & x_{24}\\ \hdashline
	x_{31} & x_{41} & x_{33} & x_{43}\\
	x_{32} & x_{42} & x_{34} & x_{44}
\end{array}
\right].
\]
As stated above, we can express the dual cone $\mathrm{SEP}^*(A;B)$ explicitly. 
Indeed, the combination of \cite{Horodecki:1996} and \cite{Lewenstein.et.al:2000} implies the following proposition.

\begin{prop}\label{prop:Horodecki}
	If $(\dim\cH_A, \dim\cH_B)=(2,2)$, then 
	\[
	\mathrm{SEP}^*(A;B) = \set{T+\Gamma(T') | T,T'\in\cT_+(AB)}.
	\]
\end{prop}

Next, we give an example of two pure states 
that are perfectly distinguishable in SEP despite being non-orthogonal. 
What follows is also a special case of our main result.

\begin{exam}[Perfect discrimination of non-orthogonal pure states in SEP]\label{exam}  
Suppose that two pure states $\rho_1,\rho_2\in\mathrm{SEP}(A;B)$ are given as 
\begin{equation}
	\rho_1 = 
	\begin{bmatrix}
		1 & 0\\
		0 & 0
	\end{bmatrix}
	\otimes
	\begin{bmatrix}
		1 & 0\\
		0 & 0
	\end{bmatrix}
	,\quad \rho_2 = 
	\begin{bmatrix}
		1-\alpha_1 & \beta_1\\
		\beta_1 & \alpha_1
	\end{bmatrix}
	\otimes
	\begin{bmatrix}
		1-\alpha_2 & \beta_2\\
		\beta_2 & \alpha_2
	\end{bmatrix}
	,\label{state:ex}
\end{equation}
where $\alpha_i\in[0,1]$, $\beta_i\ge0$, and $\beta_i^2 = \alpha_i(1-\alpha_i)$ for all $i=1,2$. 
Assume $\alpha_1+\alpha_2=1$ here. 
Then we show that $\rho_1$ and $\rho_2$ are perfectly distinguishable in SEP. 
Let us give a measurement $\{ T_1+\Gamma(T_1), T_2+\Gamma(T_2) \}$ with positive semi-definite matrices $T_1$ and $T_2$. 
Since $T_1$ and $T_2$ are positive semi-definite, 
proposition~\ref{prop:Horodecki} implies that $T_i+\Gamma(T_i)\in\mathrm{SEP}^*(A;B)$ for all $i=1,2$. 
Now, we set the positive semi-definite matrices $T_1$ and $T_2$ as 
\begin{equation*}
	T_1 = \frac{1}{2}
	\begin{bmatrix}
		1 & 0 & 0 & -1\\
		0 & 0 & 0 & 0\\
		0 & 0 & 0 & 0\\
		-1 & 0 & 0 & 1
	\end{bmatrix}
	,\quad
	T_2 = \frac{1}{2}
	\begin{bmatrix}
		0 & 0 & 0 & 0\\
		0 & 1 & 1 & 0\\
		0 & 1 & 1 & 0\\
		0 & 0 & 0 & 0
	\end{bmatrix}.
\end{equation*}
Then $\{ T_1+\Gamma(T_1), T_2+\Gamma(T_2) \}$ is a measurement of SEP 
because $(T_1+\Gamma(T_1)) + (T_2+\Gamma(T_2))=I$. 
The measurement $\{ T_1+\Gamma(T_1), T_2+\Gamma(T_2) \}$ discriminates $\rho_1$ and $\rho_2$ perfectly.
Let us verify it. 
First, the equation $\Tr\rho_1(T_2+\Gamma(T_2))=0$ follows from the definitions. 
Next, note that the assumption $\alpha_1+\alpha_2=1$ implies $\beta_1=\beta_2=\sqrt{\alpha_1\alpha_2}$. 
Since (i) $\Gamma(\rho_2)=\rho_2$ and (ii) $\alpha_1+\alpha_2=1$ ($\beta_1=\beta_2=\sqrt{\alpha_1\alpha_2}$), 
we have 
\[
\Tr\rho_2(T_1+\Gamma(T_1))
\overset{\text{(i)}}{=} 2\Tr\rho_2 T_1
= (1-\alpha_1)(1-\alpha_2) + \alpha_1\alpha_2 - 2\beta_1\beta_2
\overset{\text{(ii)}}{=} 0.
\]
Thus the equation $\Tr\rho_2(T_1+\Gamma(T_1))=0$ also follows. 
Finally, the equation $\Tr\rho_i(T_i+\Gamma(T_i))=1$ follows from $(T_1+\Gamma(T_1)) + (T_2+\Gamma(T_2))=I$ and $\Tr\rho_i(T_j+\Gamma(T_j))=0$ for all $i\neq j$. 
Therefore, the measurement $\{ T_1+\Gamma(T_1), T_2+\Gamma(T_2) \}$ discriminates $\rho_1$ and $\rho_2$ perfectly. 
Here, note that $\rho_1$ and $\rho_2$ are not orthogonal if $\alpha_1,\alpha_2\not=1$. 
Thus perfect discrimination of two pure states in SEP is possible 
even when the two states are not orthogonal.

\begin{figure}[t]
	\begin{minipage}{0.5\hsize}
		\centering
		\includegraphics[scale=1.2]{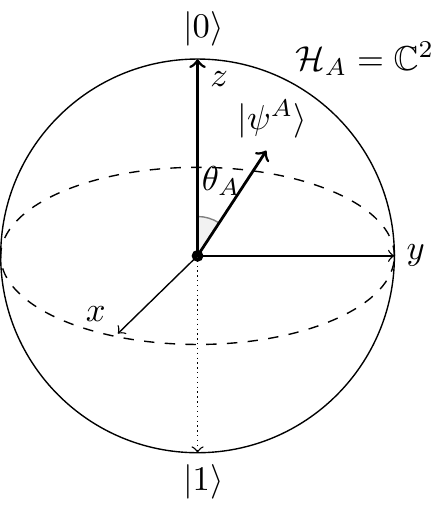}
	\end{minipage}
	{\huge $\otimes$}
	\begin{minipage}{0.5\hsize}
		\centering
		\includegraphics[scale=1.2]{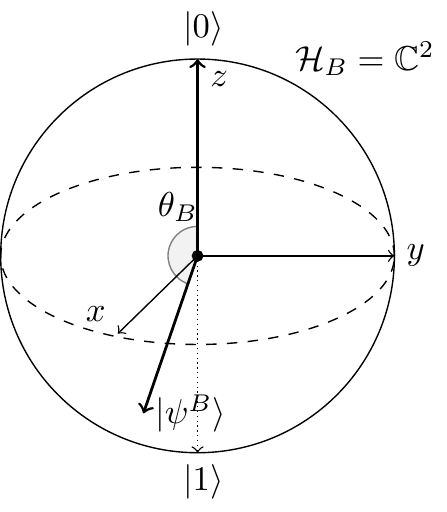}
	\end{minipage}
	\caption{The figure illustrates the two unit vectors $\ket{0}\otimes\ket{0}$ and $\ket{\psi^A}\otimes\ket{\psi^B}$ 
	in two qubits by using the two Bloch spheres. 
	Example~\ref{exam} shows that the two pure states 
	$\ketbra{0}{0}\otimes\ketbra{0}{0}$ and $\ketbra{\psi^A}{\psi^A}\otimes\ketbra{\psi^B}{\psi^B}$ 
	are perfectly distinguishable in SEP if $\theta_A+\theta_B=\pi$.}
	\label{fig:exam}
\end{figure}

Figure~\ref{fig:exam} illustrates this example by using the two Bloch spheres. 
Let $\{\ket{0},\ket{1}\}$ be an orthonormal basis of a qubit. 
Then the state $\rho_1$ can be expressed as $\rho_1=\ketbra{0}{0}\otimes\ketbra{0}{0}$. 
Since $\rho_2$ is also a separable pure state, 
there exist two unit vectors $\ket{\psi^A}$ and $\ket{\psi^B}$ such that 
$\rho_2=\ketbra{\psi^A}{\psi^A}\otimes\ketbra{\psi^B}{\psi^B}$. 
The condition $\alpha_1+\alpha_2=1$ given above corresponds to 
the condition $\theta_A+\theta_B=\pi$ of the angles in figure~\ref{fig:exam}.
\end{exam}

Example~\ref{exam} gives a sufficient condition of perfect discrimination, 
but it does not give a necessary condition. 
Thus we give the following theorem as a necessary and sufficient condition 
for two pure states to be discriminated perfectly.

\begin{theo}\label{theo:Main}
	Two pure states $\rho_1=\rho_1^A\otimes\rho_1^B$ and $\rho_2=\rho_2^A\otimes\rho_2^B$ are
	perfectly distinguishable in SEP if and only if 
	\[
	\Tr\rho_1^A\rho_2^A + \Tr \rho_1^B\rho_2^B \le 1.
	\]
\end{theo}

Here, let us compare the necessary and sufficient condition in SEP with that in quantum theory. 
In quantum theory, the condition $(\Tr\rho_1^A\rho_2^A)(\Tr\rho_1^B\rho_2^B)=0$ 
is necessary and sufficient to discriminate the two state in theorem~\ref{theo:Main} perfectly. 
Thus we can find that measurements of SEP improve the performance of state discrimination. 
The sufficiency of theorem~\ref{theo:Main} is proved in section~\ref{sect3} 
and the necessity of theorem~\ref{theo:Main} is proved in section~\ref{sect4}.

Measurements of SEP improve the performance of multiple-copy state discrimination more dramatically. 
To see this fact, let us consider perfect discrimination of $2n$-copies 
$\rho_1^{\otimes2n}$ and $\rho_2^{\otimes2n}$ of pure states. 
Then $\rho_i^{\otimes2n} = \rho_i^{\otimes n}\otimes\rho_i^{\otimes n}$ 
is a separable pure state on a bipartite system for $i=1,2$. 
Thus $\rho_1^{\otimes2n}$ and $\rho_2^{\otimes2n}$ are perfectly distinguishable in SEP 
if 
\[
2(\Tr \rho_1\rho_2)^n = \Tr \rho_1^{\otimes n}\rho_2^{\otimes n} + \Tr \rho_1^{\otimes n}\rho_2^{\otimes n} \le 1.
\] 
This inequality always holds for a sufficiently large $n$ if $\rho_1\neq\rho_2$. 
Therefore, $\rho_1^{\otimes2n}$ and $\rho_2^{\otimes2n}$ are perfectly distinguishable in SEP. 
Of course, such a measurement to realize the above perfect discrimination is impossible in quantum theory.

Next, we discuss how many states are simultaneously and perfectly distinguishable in SEP. 
That is, our interest is the \textit{capacity} $N_{\mathrm{SEP}}$ defined as 
the maximum number of simultaneously and perfectly distinguishable states in SEP: 
\[
N_{\mathrm{SEP}}
\coloneqq \max\set{n\in\mathbb{N} | \exists\{\rho_i\}_{i=1}^n,\ \exists\{M_j\}_{j=1}^n\ \text{s.t.}\ \Tr\rho_iM_j=\delta_{ij}},
\]
where $\{\rho_i\}_{i=1}^n$ and $\{M_j\}_{j=1}^n$ are a family of states in SEP and a measurement of SEP, respectively. 
As stated in the previous paragraph, 
the performance of state discrimination in SEP is higher than that in quantum theory. 
Hence one might guess that the capacity in SEP is greater than that in quantum theory. 
However, the following proposition shows that this is not the case.

\begin{prop}\label{theo:capa}
	The capacity $N_{\mathrm{SEP}}$ is $\dim(\cH_A\otimes\cH_B)$.
\end{prop}

Since the capacity in quantum theory is equal to the dimension of a quantum system,
proposition~\ref{theo:capa} asserts that $\mathrm{SEP}$ has the same capacity as quantum theory.
Actually, proposition~\ref{theo:capa} follows from \cite[lemma~24 (iii)]{Muller.et.al:2012} 
which is a more general statement on capacities
of composite systems of two quantum subsystems.\footnote{The statement \cite[theorem~3]{Masanes:2011} 
is also the same statement on capacities as \cite[lemma~24 (iii)]{Muller.et.al:2012}. 
However, it assumes the additional requirement \cite[requirement~3]{Masanes:2011} 
that all systems of the same type with the same capacity are equivalent up to invertible linear transformation,
and it is not clear that SEP satisfies the requirement.
Therefore, we do not use \cite[theorem~3]{Masanes:2011} here.}
To use \cite[lemma~24 (iii)]{Muller.et.al:2012}, 
we need to verify the \textit{transitivity} of quantum theory on $\cH_A$ ($\cH_B$)
and SEP on $\cH_A\otimes \cH_B$.
When we give a group $\mathcal{G}$ of transformations mapping states to states in quantum theory (resp.\ SEP), 
\textit{transitivity} asserts that, for any pair of two pure states $\rho_1$ and $\rho_2$ in quantum theory (resp.\ SEP), 
there exists a transformation $T\in\mathcal{G}$ such that $\rho_1=T\rho_2$.
For $X=A,B$, quantum theory on $\cH_X$ is transitive because the group
\[
\mathcal{G}_{\mathrm{QT},\,X} \coloneqq \set{\rho \mapsto U^\dag \rho U | \text{$U$ unitary on $\cH_X$}}
\]
satisfies the assertion of transitivity. 
Also, SEP on the composite system $\cH_A\otimes \cH_B$ is also transitive because the group 
\[
\mathcal{G}_{\mathrm{SEP},\,A;B}
\coloneqq \set{\rho \mapsto (U_A^\dag\otimes U_B^\dag) \rho (U_A\otimes U_B) | 
\text{$U_A$, $U_B$ unitary on $\cH_A$ and $\cH_B$}}
\]
satisfies the assertion of transitivity.
Additionally, we need the \textit{maximally mixed condition}:
for each system, the average $\int T(\rho)\,dT$ is equal to the state $I/D$  for any pure state $\rho$,
where $dT$ is the Haar measure on $\mathcal{G}$ and $D$ is the dimension of the system.
To use \cite[lemma~24 (iii)]{Muller.et.al:2012}, 
the groups $\mathcal{G}_{\mathrm{QT},\,A}$, $\mathcal{G}_{\mathrm{QT},\,B}$, and $\mathcal{G}_{\mathrm{SEP},\,A;B}$
need to satisfy the maximally mixed condition.
Fortunately, this is indeed the case.
Therefore, SEP on the composite system $\cH_A\otimes \cH_B$
satisfies the assumption of \cite[lemma~24 (iii)]{Muller.et.al:2012} and thus 
proposition~\ref{theo:capa} follows.

Table~\ref{table} summarizes the necessary and sufficient conditions of perfect discrimination and the capacities in quantum theory and SEP. 
The performance of perfect discrimination in SEP is better than that in quantum theory 
but the capacity in SEP is equal to that in quantum theory.

\begin{table}[t]
	\centering
	\caption{Necessary and sufficient conditions of perfect discrimination of two pure states, and capacities.}\label{table}
	\begin{tabular}{|c||c|c|}\hline
		GPTs&SEP&Quantum theory\\ \hline
		Necessary and sufficient&\multirow{3}{*}{$\Tr\rho_1^A\rho_2^A+\Tr\rho_1^B\rho_2^B\leq1$}&\multirow{3}{*}{$(\Tr\rho_1^A\rho_2^A)(\Tr\rho_1^B\rho_2^B)=0$}\\
		condition&   &    \\
		of perfect discrimination&   &    \\ \hline
		Capacity&$\dim(\cH_A\otimes \cH_B)$&$\dim(\cH_A\otimes \cH_B)$\\ \hline
	\end{tabular}
\end{table}

\section{Proof of the sufficiency of theorem~\ref{theo:Main}}\label{sect3}

In this section, we prove the sufficiency of theorem~\ref{theo:Main}. 
Any pair of two pure states $\rho_1$ and $\rho_2$ in SEP can be expressed as 
$\rho_i = \ketbra{u_i^A}{u_i^A}\otimes\ketbra{u_i^B}{u_i^B}$ 
by using unit vectors $u_1^A,u_2^A\in\cH_A$ and $u_1^B,u_2^B\in\cH_B$. 
Thus there exist two-dimensional subspaces 
$\cH_{A'}\subset\cH_A$ and $\cH_{B'}\subset\cH_B$ 
such that $\rho_1,\rho_2\in\mathrm{SEP}(A';B')$. 
Hence, for all integers $d_A,d_B\ge2$, 
the $(d_A,d_B)$-dimensional case can be reduced to the $(2,2)$-dimensional case. 
Without loss of generality, the above states $\rho_1$ and $\rho_2$ can be written as \eqref{state:ex}. 
After all, it suffices to prove the following theorem.

\begin{theo}[Sufficiency of theorem~\ref{theo:Main}]\label{theo:Suf}
	Assume that two pure states in SEP are given as \eqref{state:ex}. 
	If $\alpha_1+\alpha_2 \ge 1$, then $\rho_1$ and $\rho_2$ are perfectly distinguishable in SEP.
\end{theo}

We prove theorem~\ref{theo:Suf} by giving measurements of SEP explicitly. 
It is the most difficult point in the proof of theorem~\ref{theo:Suf} 
to prove that elements of measurements belong to $\mathrm{SEP}^\ast(A;B)$. 
Since proposition~\ref{prop:Horodecki} can be applied to the $(2,2)$-dimensional case, 
elements of measurements are give as the form in proposition~\ref{prop:Horodecki}. 
Thus the proof of theorem~\ref{theo:Suf} requires us to prove that certain matrices are positive semi-definite. 
To prove positive semi-definiteness, we need to investigate principal submatrices. 
For a matrix $X=(x_{ij})_{1\le i,j\le d}$ and a set $J\subset\{1,2,\ldots,d\}$, 
the principal submatrix $X(J)$ is defined as the matrix $(x_{ij})_{i,j\in J}$. 
As a criterion of positive semi-definiteness, 
the following proposition is well-known.

\begin{prop}[Section~7 in \cite{Horn.Jonson.Book:2012}]\label{prop:PosCrite}
	Let $X$ be a Hermitian matrix whose rank is $r$. 
	If any integer $1\le k\le r$ satisfies $\det X(1,2,\ldots,k)>0$, 
	then $X$ is positive semi-definite.
\end{prop}

Having fixed the notation according to the above explanation, let us prove theorem~$\ref{theo:Suf}$.

\begin{proof}[Proof of theorem~$\ref{theo:Suf}$]
The case $\alpha_1+\alpha_2=1$ has been already proved in example~\ref{exam}. 
Thus we assume $\gamma \coloneqq \alpha_1+\alpha_2 > 1$ in this proof. 
Then the condition $\alpha_i\neq0$ holds due to $\alpha_i\in[0,1]$ and $\gamma = \alpha_1+\alpha_2 > 1$. 
Now, we define the matrices $T_1$ and $T_2$ as 
\begin{align*}
	T_1 &= \frac{1}{2\gamma}
	\begin{bmatrix}
		\gamma & 0 & 0 & -\beta_1\beta_2\gamma/\alpha_1\alpha_2\\
		0 & \gamma-1 & 0 & -(\gamma-1)\beta_1/\alpha_1\\
		0 & 0 & \gamma-1 & -(\gamma-1)\beta_2/\alpha_2\\
		-\beta_1\beta_2\gamma/\alpha_1\alpha_2 &
		-(\gamma-1)\beta_1/\alpha_1 & -(\gamma-1)\beta_2/\alpha_2 &
		 2-\gamma
	\end{bmatrix}
	,\\
	T_2 &= \frac{1}{2\gamma}
	\begin{bmatrix}
		\qquad 0\qquad&\qquad 0\qquad&\qquad 0\qquad&\qquad 0\\
		\qquad 0\qquad&\qquad 1\qquad&\quad \beta_1\beta_2\gamma/\alpha_1\alpha_2 \quad&\quad (\gamma-1)\beta_1/\alpha_1\\
		\qquad 0\qquad&\quad \beta_1\beta_2\gamma/\alpha_1\alpha_2 \quad&\qquad 1\qquad&\quad (\gamma-1)\beta_2/\alpha_2\\
		\qquad 0\qquad&\quad (\gamma-1)\beta_1/\alpha_1 \quad&\quad (\gamma-1)\beta_2/\alpha_2 \quad&\quad 2(\gamma-1)
	\end{bmatrix}.
\end{align*}
Let us show that $\{T_1+\Gamma(T_1), T_2+\Gamma(T_2)\}$ is a measurement of SEP and discriminates $\rho_1$ and $\rho_2$ perfectly. 
That is, let us verify that 
\begin{gather}
	T_1+\Gamma(T_1) + T_2+\Gamma(T_2) = I,\label{eqY01}\\
	\Tr \rho_i(T_j+\Gamma(T_j)) = \delta_{ij} \quad (\forall i,j=1,2),\label{eqY02}\\
	T_i+\Gamma(T_i)\in\mathrm{SEP}^*(A;B) \quad (\forall i=1,2)\label{eqY03}.
\end{gather}
The equation \eqref{eqY01} follows from the definitions of $T_1$ and $T_2$. 
The equation \eqref{eqY01} and the invariance of $\rho_1,\rho_2$ under $\Gamma$ 
reduce \eqref{eqY02} to $\Tr\rho_2T_1=0$ and $\Tr\rho_1T_2=0$. 
Since the equation $\Tr\rho_1T_2=0$ follows from the definitions of $\rho_1$ and $T_2$, 
the remaining of \eqref{eqY02} is $\Tr\rho_2T_1=0$. 
Define the two unit vectors $\ket{\psi^A}$ and $\ket{\psi^B}$ as 
\[
\ket{\psi^A} = [\sqrt{1-\alpha_1}, \sqrt{\alpha_1}]^\top,\quad
\ket{\psi^B} = [\sqrt{1-\alpha_2}, \sqrt{\alpha_2}]^\top.
\]
Then $\rho_2=\ketbra{\psi^A}{\psi^A}\otimes\ketbra{\psi^B}{\psi^B}$. 
Noting $\gamma=\alpha_1+\alpha_2$, $\beta_i=\sqrt{\alpha_i(1-\alpha_i)}$, and 
\begin{equation*}
	\ket{\psi^A\otimes\psi^B}
	= \Bigl[ \sqrt{(1-\alpha_1)(1-\alpha_2)},
	\sqrt{(1-\alpha_1)\alpha_2},
	\sqrt{\alpha_1(1-\alpha_2)},
	\sqrt{\alpha_1\alpha_2} \Bigr]^\top,
\end{equation*}
we find that all entries of $2\gamma T_1\ket{\psi^A\otimes\psi^B}$ vanish: 
\begin{align*}
	\text{The first entry}
	&= \gamma\sqrt{(1-\alpha_1)(1-\alpha_2)} - (\beta_1\beta_2\gamma/\alpha_1\alpha_2)\sqrt{\alpha_1\alpha_2}\\
	&= \gamma\sqrt{(1-\alpha_1)(1-\alpha_2)} - \sqrt{(1-\alpha_1)(1-\alpha_2)}\gamma = 0,\\
	\text{The second entry}
	&= (\gamma-1)\sqrt{(1-\alpha_1)\alpha_2} - (\gamma-1)(\beta_1/\alpha_1)\sqrt{\alpha_1\alpha_2}\\
	&= (\gamma-1)\sqrt{(1-\alpha_1)\alpha_2} - (\gamma-1)\sqrt{(1-\alpha_1)\alpha_2} = 0,\\
	\text{The third entry}
	&= (\gamma-1)\sqrt{\alpha_1(1-\alpha_2)} - (\gamma-1)(\beta_2/\alpha_2)\sqrt{\alpha_1\alpha_2}\\
	&= (\gamma-1)\sqrt{\alpha_1(1-\alpha_2)} - (\gamma-1)\sqrt{(1-\alpha_2)\alpha_1} = 0,\\
	\text{The fourth entry}
	&= -(\beta_1\beta_2\gamma/\alpha_1\alpha_2)\sqrt{(1-\alpha_1)(1-\alpha_2)}
	- (\gamma-1)(\beta_1/\alpha_1)\sqrt{(1-\alpha_1)\alpha_2}\\
	&\quad- (\gamma-1)(\beta_2/\alpha_2)\sqrt{\alpha_1(1-\alpha_2)}
	+ (2-\gamma)\sqrt{\alpha_1\alpha_2}\\
	&= -\gamma(1-\alpha_1)(1-\alpha_2)/\sqrt{\alpha_1\alpha_2}
	- (\gamma-1)(1-\alpha_1)\sqrt{\alpha_2/\alpha_1}\\
	&\quad- (\gamma-1)(1-\alpha_2)\sqrt{\alpha_1/\alpha_2}
	+ (2-\gamma)\sqrt{\alpha_1\alpha_2}\\
	&= \bigl\{ -\gamma(1-\alpha_1)(1-\alpha_2) - (\gamma-1)(1-\alpha_1)\alpha_2\\
	&\qquad- (\gamma-1)(1-\alpha_2)\alpha_1
	+ (2-\gamma)\alpha_1\alpha_2 \bigr\}/\sqrt{\alpha_1\alpha_2}\\
	&= \bigl\{ -\gamma\bigl( (1-\alpha_1) + \alpha_1 \bigr)\bigl( (1-\alpha_2) + \alpha_2 \bigr)\\
	&\qquad+ (1-\alpha_1)\alpha_2 + (1-\alpha_2)\alpha_1 + 2\alpha_1\alpha_2 \bigr\}/\sqrt{\alpha_1\alpha_2}\\
	&= (-\gamma + \alpha_1 + \alpha_2)/\sqrt{\alpha_1\alpha_2} = 0.
\end{align*}
The above fact implies $\Tr\rho_2 T_1 = \braket{\psi^A\otimes\psi^B |T_1| \psi^A\otimes\psi^B} = 0$. 
Therefore, \eqref{eqY02} holds.
\par
Finally, we verify \eqref{eqY03}. 
As already stated, \eqref{eqY03} follows from the positive semi-definiteness of $T_1$ and $T_2$. 
Since $T_1\ket{\psi^A\otimes\psi^B}=0$, the rank of $T_1$ is at most three. 
Moreover, the determinants of $T_1(1)$, $T_1(1,2)$, and $T_1(1,2,3)$ are positive due to the assumption $1<\gamma\le2$. 
Thus proposition~\ref{prop:PosCrite} implies that $T_1$ is positive semi-definite. 
The remaining of \eqref{eqY03} is the positive semi-definiteness of $T_2$. 
To prove it, we verify the positive definiteness of $T_2(2,3,4)$. 
First, the inequality of arithmetic and geometric means yields 
(i) $\gamma^2 = (\alpha_1+\alpha_2)^2 \ge 4\alpha_1\alpha_2$. 
The inequality (i) and the assumption (ii) $1<\gamma\le2$ imply 
\begin{align*}
	&\quad (2\gamma)^2\det T_2(2,3)
	= \det
	\begin{bmatrix}
		1 & \beta_1\beta_2\gamma/\alpha_1\alpha_2\\
		\beta_1\beta_2\gamma/\alpha_1\alpha_2 & 1
	\end{bmatrix}
	= 1 - (\beta_1\beta_2\gamma/\alpha_1\alpha_2)^2\\
	&= 1 - \frac{(1-\alpha_1)(1-\alpha_2)\gamma^2}{\alpha_1\alpha_2}
	= 1 - \frac{(1-\gamma+\alpha_1\alpha_2)\gamma^2}{\alpha_1\alpha_2}
	= 1-\gamma^2 + \frac{(\gamma-1)\gamma^2}{\alpha_1\alpha_2}\\
	&= (\gamma-1)\Bigl( -(1+\gamma) + \frac{\gamma^2}{\alpha_1\alpha_2} \Bigr)
	\overset{\text{(i)}}{\ge} (\gamma-1)(-1-\gamma+4) \overset{\text{(ii)}}{>} 0.
\end{align*}
Second, the determinant of $T_2(2,3,4)$ can be calculated as follows: 
\begin{align*}
	&\quad \frac{(2\gamma)^3}{\gamma-1}\det T_2(2,3,4)
	= \det
	\begin{bmatrix}
		1 & \beta_1\beta_2\gamma/\alpha_1\alpha_2 & \beta_1/\alpha_1\\
		\beta_1\beta_2\gamma/\alpha_1\alpha_2 & 1 & \beta_2/\alpha_2\\
		(\gamma-1)\beta_1/\alpha_1 & (\gamma-1)\beta_2/\alpha_2 & 2
	\end{bmatrix}
	\\
	&= 2 + 2\gamma(\bcancel{\gamma} - 1)\beta_1^2\beta_2^2/\alpha_1^2\alpha_2^2
	- (\gamma-1)\beta_1^2/\alpha_1^2 - \bcancel{2\gamma^2\beta_1^2\beta_2^2/\alpha_1^2\alpha_2^2}
	- (\gamma-1)\beta_2^2/\alpha_2^2\\
	&= 2 - \frac{2\gamma(1-\alpha_1)(1-\alpha_2)}{\alpha_1\alpha_2}
	- \frac{(\gamma-1)\{ (1-\alpha_1)\alpha_2 + \alpha_1(1-\alpha_2)\}}{\alpha_1\alpha_2}\\
	&= \bcancel{2} - \frac{2\gamma(1-\gamma + \bcancel{\alpha_1\alpha_2})}{\alpha_1\alpha_2}
	- \frac{(\gamma-1)(\gamma - \bcancel{2\alpha_1\alpha_2})}{\alpha_1\alpha_2}\\
	&= \frac{2\gamma(\gamma-1)}{\alpha_1\alpha_2} - \frac{(\gamma-1)\gamma}{\alpha_1\alpha_2}
	= \frac{\gamma(\gamma-1)}{\alpha_1\alpha_2} \overset{\text{(ii)}}{>} 0.
\end{align*}
Since the determinants of $T_2(2)$, $T_2(2,3)$, and $T_2(2,3,4)$ are positive, 
proposition~\ref{prop:PosCrite} implies that $T_2(2,3,4)$ is positive semi-definite. 
Therefore, $T_2$ is also positive semi-definite and then we finish this proof.
\end{proof}

\section{Proof of the necessity of theorem~\ref{theo:Main}}\label{sect4}

In this section, we prove the necessity of theorem~\ref{theo:Main}. 
As stated in section~\ref{sect3}, for all integers $d_A,d_B\ge2$, 
the $(d_A,d_B)$-dimensional case can be reduced to the $(2,2)$-dimensional case. 
Although $\rho_1$ and $\rho_2$ in theorem~\ref{theo:Main} are pure states, 
we address not necessarily pure product states in the $(2,2)$-dimensional case. 
Hence, without loss of generality, it suffices to prove the following theorem.

\begin{theo}[Necessity of theorem~\ref{theo:Main}]\label{theo:Nec}
	Assume that two product states $\rho_1$ and $\rho_2$ in SEP are given as 
	\begin{equation}
		\rho_1 = 
		\begin{bmatrix}
			1-p_1 & 0\\
			0 & p_1
		\end{bmatrix}
		\otimes
		\begin{bmatrix}
			1-p_2 & 0\\
			0 & p_2
		\end{bmatrix}
		,\quad \rho_2 = 
		\begin{bmatrix}
			1-\alpha_1 & \beta_1\\
			\beta_1 & \alpha_1
		\end{bmatrix}
		\otimes
		\begin{bmatrix}
			1-\alpha_2 & \beta_2\\
			\beta_2 & \alpha_2
		\end{bmatrix}
		,\label{state}
	\end{equation}
	where $\alpha_i,p_i\in[0,1]$, $\beta_i\ge0$, and $\beta_i^2 \le \alpha_i(1-\alpha_i)$ for all $i=1,2$.
	If $\rho_1$ and $\rho_2$ are perfectly distinguishable in SEP, then 
	\begin{equation}
		\Tr\rho_1\rho_2 \le \beta_1\beta_2\abs{(2p_1-1)(2p_2-1)} \le \beta_1\beta_2 \le 1/4. \label{eq:mix}
	\end{equation}
	If the additional condition $p_1=p_2=0$ holds, then \eqref{eq:mix} can be reduced to 
	\begin{equation}
		\alpha_1 + \alpha_2 \ge 1. \label{eq:pure}
	\end{equation}
\end{theo}

If the additional condition $p_1=p_2=0$ holds, 
the equation \eqref{state} turns to \eqref{state:ex}. 
Hence \eqref{eq:pure} means the necessity of theorem~\ref{theo:Main}. 
Since theorem~\ref{theo:Nec} handles mixed states, 
another additional condition leads to a result which theorem~\ref{theo:Main} does not imply. 
For instance, theorem~\ref{theo:Main} does not imply the following: 
if the additional condition $\beta_1=0$ holds, 
then \eqref{eq:mix} implies that $\rho_1$ and $\rho_2$ are orthogonal. 
Thus, if $\rho_1^A$ and $\rho_2^A$ are diagonal matrices, 
the perfect discrimination of $\rho_1=\rho_1^A\otimes\rho_1^B$ and $\rho_2=\rho_2^A\otimes\rho_2^B$ in SEP 
implies the orthogonality of $\rho_1$ and $\rho_2$.

\begin{proof}[Proof of theorem~$\ref{theo:Nec}$.]
Proposition~\ref{prop:Horodecki} implies that 
any element of $\mathrm{SEP}^\ast(A;B)$ can be expressed as $S+\Gamma(S')$ 
with positive semi-definite matrices $S$ and $S'$. 
Thus we assume that a measurement $\{S_1+\Gamma(S'_1), S_2+\Gamma(S'_2)\}$ 
with positive semi-definite matrices $S_1$, $S'_1$, $S_2$ and $S'_2$ 
discriminates $\rho_1$ and $\rho_2$ perfectly in SEP. 
That is, the equation 
\[
\Tr\rho_1(S_2+\Gamma(S'_2)) = \Tr\rho_2(S_1+\Gamma(S'_1)) = 0
\]
holds. Since $\Gamma$ fixes $\rho_1$ and $\rho_2$, 
the equation 
\[
\Tr\rho_1(\Gamma(S_2)+S'_2) = \Tr\rho_2(\Gamma(S_1)+S'_1) = 0
\]
also holds. Put $T_i=(S_i+S'_i)/2$ for all $i=1,2$. 
Then $\{T_1+\Gamma(T_1), T_2+\Gamma(T_2)\}$ is a measurement of SEP 
and satisfies 
\begin{equation*}
	\Tr\rho_1(T_2+\Gamma(T_2)) = \Tr\rho_2(T_1+\Gamma(T_1)) = 0.
\end{equation*}
Since $\Gamma$ fixes $\rho_1$ and $\rho_2$, 
the above equation can be reduced to 
\begin{equation}
	\Tr\rho_1 T_2 = \Tr\rho_2 T_1 = 0. \label{eqY1}
\end{equation}
Furthermore, the matrix $T \coloneqq T_1+T_2$ satisfies 
\begin{equation*}
	T + \Gamma(T) = I.
\end{equation*}
Thus the positive semi-definite matrix $T$ can be written as 
\begin{align*}
	T &= 
	\left[
	\begin{array}{cccc}
		1/2&-ix_1&0&-z\\
		ix_1&1/2&z&0\\ 
		0&\overline{z}&1/2&-ix_2\\
		-\overline{z}&0&ix_2&1/2
	\end{array}
	\right]\\
	&= \frac{1}{2}I + 
	\begin{bmatrix}
		x_1 & 0\\
		0 & x_2
	\end{bmatrix}
	\otimes
	\begin{bmatrix}
		0 & -i\\
		i & 0
	\end{bmatrix}
	+
	\begin{bmatrix}
		0 & -z\\
		\overline{z} & 0
	\end{bmatrix}
	\otimes
	\begin{bmatrix}
		0 & 1\\
		-1 & 0
	\end{bmatrix}
\end{align*}
for some $x_1,x_2\in\mathbb{R}$ and $z\in\mathbb{C}$. 
From the positive semi-definiteness of $T$, the inequality $|z|\le1/2$ follows.
\par
Now, we show \eqref{eq:mix}. 
Since \eqref{eqY1} implies $\rho_1 T\rho_2 = \rho_1(T_1+T_2)\rho_2 = 0$, 
we obtain 
\[
\Tr\rho_1\rho_2 = \Tr\rho_1(T + \Gamma(T))\rho_2
= \Tr\rho_1\Gamma(T)\rho_2 = \Tr\Gamma(T)\rho_2\rho_1.
\]
Hereinafter, $\Re X$ denotes the real part of a matrix $X$. 
Then the matrices $\rho_2\rho_1$ and $\Re\Gamma(T)$ are calculated as follows: 
\begin{gather*}
	\begin{split}
		\rho_2\rho_1 &= 
		\begin{bmatrix}
			1-\alpha_1 & \beta_1\\
			\beta_1 & \alpha_1
		\end{bmatrix}
		\begin{bmatrix}
			1-p_1 & 0\\
			0 & p_1
		\end{bmatrix}
		\otimes
		\begin{bmatrix}
			1-\alpha_2 & \beta_2\\
			\beta_2 & \alpha_2
		\end{bmatrix}
		\begin{bmatrix}
			1-p_2 & 0\\
			0 & p_2
		\end{bmatrix}
		\\
		&= 
		\begin{bmatrix}
			(1-\alpha_1)(1-p_1) & \beta_1 p_1\\
			\beta_1(1-p_1) & \alpha_1 p_1
		\end{bmatrix}
		\otimes
		\begin{bmatrix}
			(1-\alpha_2)(1-p_2) & \beta_2 p_2\\
			\beta_2(1-p_2) & \alpha_2 p_2
		\end{bmatrix},
	\end{split}
	\\
	\Re\Gamma(T) = \frac{1}{2}I + 
	\begin{bmatrix}
		0 & -\Re z\\
		\Re z & 0
	\end{bmatrix}
	\otimes
	\begin{bmatrix}
		0 & -1\\
		1 & 0
	\end{bmatrix}
	= \frac{1}{2}I + (\Re z)
	\begin{bmatrix}
		0 & -1\\
		1 & 0
	\end{bmatrix}
	^{\otimes2}.
\end{gather*}
Since all entries of $\rho_2\rho_1$ and $\Tr \rho_1\rho_2$ are real, 
the equation 
\[
\Tr\rho_1\rho_2 = \Tr\Gamma(T)\rho_2\rho_1 = \Re\Tr\Gamma(T)\rho_2\rho_1 = \Tr\Re(T(\Gamma))\rho_2\rho_1
\]
holds. Thus it follows that 
\begin{align}
	\Tr\rho_1\rho_2
	&= \Tr\Re(T(\Gamma))\rho_2\rho_1\nonumber\\
	&= \frac{1}{2}\Tr\rho_2\rho_1 + (\Re z)\Tr
	\begin{bmatrix}
		0 & -1\\
		1 & 0
	\end{bmatrix}
	\begin{bmatrix}
		(1-\alpha_1)(1-p_1) & \beta_1 p_1\\
		\beta_1(1-p_1) & \alpha_1 p_1
	\end{bmatrix}
	\nonumber\\
	&\quad\cdot \Tr
	\begin{bmatrix}
		0 & -1\\
		1 & 0
	\end{bmatrix}
	\begin{bmatrix}
		(1-\alpha_2)(1-p_2) & \beta_2 p_2\\
		\beta_2(1-p_2) & \alpha_2 p_2
	\end{bmatrix}
	\nonumber\\
	&= \frac{1}{2}\Tr\rho_2\rho_1 + (\Re z)\beta_1(2p_1-1)\beta_2(2p_2-1),\nonumber\\
	\Tr\rho_1\rho_2 &= 2(\Re z)\beta_1\beta_2(2p_1-1)(2p_2-1). \label{eqH2}
\end{align}
Therefore, the inequality $|z|\le1/2$ and \eqref{eqH2} imply \eqref{eq:mix}.
\par
Next, assuming the additional condition $p_1=p_2=0$, we show \eqref{eq:pure}. 
The inequality $\beta_i^2 \le \alpha_i(1-\alpha_i)$ and \eqref{eq:mix} imply that 
\begin{gather*}
	\Tr\rho_1\rho_2 \le \beta_1\beta_2 \le \sqrt{\alpha_1(1-\alpha_1)\alpha_2(1-\alpha_2)}
	= \sqrt{\alpha_1\alpha_2\Tr\rho_1\rho_2},\\
	(1-\alpha)(1-\alpha_2) = \Tr\rho_1\rho_2 \le \alpha_1\alpha_2,
\end{gather*}
whence \eqref{eq:pure} holds.
\end{proof}

\section{Discussion}\label{sect5}
In this paper, we have discussed perfect discrimination in SEP, 
and has revealed that a necessary and sufficient condition to discriminate two pure states in SEP perfectly 
is that the inequality $\Tr \rho_1^A\rho_2^A+\Tr \rho_1^B\rho_2^B\leq1$ be satisfied.
More generally, let us consider perfect discrimination of two mixed states in SEP.
Two states $\rho_1,\rho_2\in\mathrm{SEP}(A;B)$ can be written as 
$\rho_1=\sum_i \alpha_i\rho_{1,i}$ and $\rho_2=\sum_j \beta_j\rho_{2,j}$,
where $\rho_{1,i}$ and $\rho_{2,j}$ are separable pure states; $\alpha_i,\beta_j>0$ and $\sum_i \alpha_i=\sum_j \beta_j=1$. 
In this case, perfect discrimination of $\rho_1$ and $\rho_2$ is equivalent to 
that of any $\rho_{1,i}$ and $\rho_{2,j}$ by a common measurement.
Hence the tuple of the above inequalities for any $\rho_{1,i}$ and $\rho_{2,j}$ 
is a necessary condition for perfect discrimination.
However, a sufficient condition must be more strict 
because the tuple of the above inequalities for any $\rho_{1,i}$ and $\rho_{2,j}$ 
does not imply the existence of a common measurement. 
Since a necessary and sufficient condition for perfect discrimination of two mixed states in SEP is not given yet,
it is a future study. 

As a related study, 
Maruyama et al.\ \cite{Maruyama.et.al:2009} pointed out the possibility of the break of the second thermodynamical law
when non-orthogonal states can be discriminated perfectly.
Since we have shown the above perfect discrimination in SEP,
it is another interesting future study to investigate the second thermodynamical law in SEP.
As another future study, 
we might apply our result to the calculation of the various type of information quantity 
defined in \cite[eqs.~(18), (26), and (29)]{Short.et.al:2010}.
This application is expected to bring us more information-theoretical study on SEP.


\ack
MH is grateful to Prof.\ Giulio Chiribella, Prof.\ Oscar Dahlsten, and Dr.\ Daniel Ebler
for helpful discussions.
He is also thankful to Mr.\ Kun Wang for his comments.
The authors are grateful to Mr.\ Seunghoan Song for providing many helpful comments for this paper. 
MH was supported in part by Japan Society for the Promotion of 
Science (JSPS) Grant-in-Aid for Scientific Research (A) No.\ 17H01280, 
(B) No.\ 16KT0017, 
and Kayamori Foundation of Informational Science Advancement. 
YY was supported by JSPS Grant-in-Aid for JSPS Fellows No.\ 19J20161.

\section*{References}

\bibliography{references}

\end{document}